\documentclass{amsart}
     
\newtheorem{theorem}{Theorem}[section]
\newtheorem{lemma}[theorem]{Lemma}

\theoremstyle{definition}

\newtheorem{example}[theorem]{Example}

\theoremstyle{remark}
\usepackage{epsfig} 
\usepackage{graphics} 
\numberwithin{equation}{section}

\def\a{{\alpha}}

\def\e{{\epsilon}}
\def\ep{{\eta}}

\def\l{{\lambda}}

\def\r{{\rho}}
\def\l{{\lambda}}
\def\m{{\mu}}

\def\p{{\psi}}

\def\ttt{{\theta}}
\def\F{{\Phi}}
\def\G{{\Gamma}}
\def\O{{\Omega}}
\def\D{{\Delta}}

\def\cc{{{\mathcal C}}}

\def\kk{{{\mathcal K}}}
\def\hh{{{\mathcal H}}}

\def\nnn{{{\bf n}}}

\def\R{{{\bf R}^1}}

\def\9{{\ \hbox{in}\ \O}}
\def\1{{\ \hbox{on}\ \G_1}}
\def\2{{\ \hbox{on}\ \G_2}}
\def\3{{\ \hbox{on}\ \G_3}}
\def\0{{\ \hbox{on}\ \G}}

\def\pa{{\partial}}
\def\pp{{\parallel}}

\begin{document}

\title[A variational principle for porous metal bearings]{A variational principle for the system of P.D.E. of  porous metal bearings}
\author{Giovanni Cimatti}
\address{Department of Mathematics, Largo Bruno
  Pontecorvo 5, 56127 Pisa Italy}
\email{cimatti@dm.unipi.it}




\dedicatory{To John Ockendon mentor and friend}


\begin{abstract}
Porous metal bearings are widely used in small and micro devices. To compute the pressure one has to solve the Reynolds equation coupled with the Laplace equation. We show that it is possible to give to the relevant boundary value problem a variational formulation. We show that  the pressure inside a porous bearing is less then that of the corresponding non-porous bearing.
\end{abstract}

\maketitle

\section{Introduction}
In a bearing made of a porous material  lubricant flow out of the bearing surface. The problem is to solve the Reynolds equation for the pressure in the oil film simultaneously with the Laplace equation for the porous matrix, the link between the two equations is represented by the continuity of the flow, see \cite{PS} page 54, \cite{C} page 548 and \cite{Ci}. In Section 2 of this paper we consider a journal bearing of finite length $L$ with the cross-section of the figure below.

\begin{figure}[h]
\centerline{\epsfig{figure=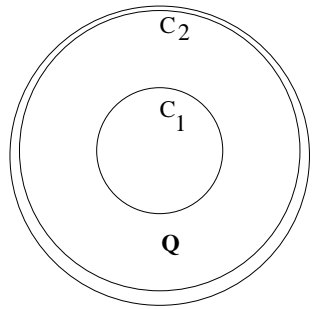,width=4cm}}
\caption{The region between the two concentric circles is the cross-section of the porous matrix}
\end{figure}

 Using cylindrical  coordinates $(\r,\ttt, y)$ and following the usual notations \cite{PS}, we define 
$h(\ttt)=c(1+\e\cos(\ttt))$ for the height of the fluid film, $(0<\e<1)$, $\m$ for the viscosity, $U$ for the velocity of the shaft and $\F$ for the permeability of the porous matrix. Let 

\begin{equation*}
Q=\{(\r,\ttt,y);\ R_1<\r<R_2,\ 0<\ttt\leq 2\pi,\ 0<y<L\}
\end{equation*}
\noindent denote the porous matrix and

\begin{equation*} 
C_1=\{(\r,\ttt,y);\ \r=R_1,\ 0<\ttt\leq 2\pi,\ 0<y<L\}
\end{equation*}

\noindent be the interior surface of the matrix and

\begin{equation*}  
C_2=\{(\r,\ttt,y);\ \r=R_2,\ 0<\ttt\leq 2\p,\ 0<y<L\}
\end{equation*}

\noindent the exterior part, where the Reynolds equation holds. Let

\begin{equation*}  
A_1=\{(\r,\ttt,y);\ R_1<\r<R_2,\ 0<\ttt\leq 2\pi,\  y=0\}
\end{equation*}

\begin{equation*}
A_2=\{(\r,\ttt,y);\ R_1<\r<R_2,\ 0<\ttt\leq 2\pi,\  y=L\}.
\end{equation*}

Define (see \cite{C}, page 549 )

\begin{equation*}
k_1=\frac{c^2}{12\Phi R_2^2},\quad k_2=\frac{c^2}{12L^2\Phi}, \quad k_3=\frac{6U\m}{12\Phi R_2}.
\end{equation*}

 The pressure $p(\r,\ttt,y)$ in the lubricating film and in the matrix is determined by the following problem

\begin{equation}
\label{1_4_1}
\D p(\r,\ttt,y)=0\quad in\  Q
\end{equation}

\begin{equation}
\label{2_4_1}
k_1\bigl[h^3(\ttt)p_\ttt(R_2,\ttt,y)\bigl]_\ttt+k_2\bigl[h^3(\ttt)p_y(R_2,\ttt,y)\bigl]_y=
\end{equation}

\begin{equation*}
p_\r(R_2,\ttt,y)-k_3 h'(\ttt)\quad on\ C_2
\end{equation*}

\begin{equation}
\label{3_4_1}
p(R_1, \ttt, y)=0\quad on\ C_1
\end{equation}

\begin{equation}
\label{4_4_1}
p(\r,\ttt,0)=0\quad on\  A_1
\end{equation}

\begin{equation}
\label{5_4_1}
p(\r,\ttt,L)=0\quad on\  A_2
\end{equation}

\begin{equation}
\label{6_4_1}
p(\r,\ttt,y)=p(\r,\ttt+2\pi,y).
\end{equation}

This problem is interesting also from the mathematical point of view for its non-standard nature. We must satisfy on part of the boundary an elliptic equation of the same order of the operator, (in this case the laplacian) valid in $Q$. It seems that this kind of problems have not been studied in the vast literature on elliptic equations of the second order.

In real bearings the lubricant film cannot sustain pressure below the atmospheric pressure \footnote{Taken here equal to zero in a suitable scale.} leading to the phenomenon of rupture in the film and cavitation. Thus the fluid film breaks down in a well-defined region. To take into account of this fact we must add to problem (\ref{1_4_1})-(\ref{6_4_1}) the unilateral condition

\begin{equation} 
\label{1_5_1}
p\geq 0\quad on\ C_2,
\end{equation}

\noindent with equation (\ref{2_4_1}) holding only where $p>0$. In this paper both problems, with and without the unilateral condition, are considered. Two approaches are possible: the so-called ``half-Sommerfeld'' condition in which the cavitated solution is simply obtained setting equal to zero the solution of the bilateral problem where it becomes negative, or the more precise point of view in which on the boundary of the region of cavitation the two conditions

\begin{equation*}
p=0\quad and\  \frac{dp}{dn}=0
\end{equation*}
 must hold. We show in Section 2 that the solution of problems (\ref{1_4_1})-(\ref{6_4_1}) or (\ref{1_4_1})-(\ref{6_4_1}) with (\ref{1_5_1}) can be obtained from a variational principle. This makes possible to use the Lax-Milgram lemma and the Stampacchia theorem \cite{B} to prove that both problems have one and only one solution.

 The special geometry of the long bearing is considered in Section 3, in this case the pressure shall depend only from two variables with the Reynolds equation becoming one-dimensional. Using the Fourier's method the two equations are uncoupled and we can solve separately a two-point problem for the Reynolds equation and a Dirichlet's problem in the porous matrix. The two-point problem is solved in a relevant special case giving a simple expression for the pressure in the film. The effect of the porous matrix reduces to a single numerical parameter. Moreover, if we adopt the half-Sommerfeld condition we find that the pressure in porous bearings is less than the one in the non-porous bearing. A variational inequality relevant in long bearings is studied in Section 4. Finally in Section 5 some open problems concerning the regularity of the solutions and the shape of the region of cavitation are proposed. We note that the shape of the region of cavitation is crucial  for the determination of the load capacity of the bearing.

\section{variational formulation}
Define the functional

\begin{equation}
\label{1_10}
J(p)=\int_0^{2\pi}\int_{R_1}^{R_2}\int_0^L\biggl[p_\r^2(\r,\ttt,y)+\frac{1}{\r^2}p_\ttt^2(\r,\ttt,y)+p_y^2(\r,\ttt,y)\biggl]\r d\ttt d\r dy+
\end{equation}

\begin{equation*}
+\int_0^{2\pi}\int_0^L\biggl\{ h^3(\ttt)\biggl[k_1 p_\ttt^2(R_2,\ttt,y)+k_2p_y^2(R_2,\ttt,y)\biggl]-2k_3 h'(\ttt)p(R_2,\ttt,y)\biggl\}R_2d\ttt dy.
\end{equation*}

Let us take as class of admissible functions for $J(p)$ the set

\begin{equation*}
\cc=\{p(\r,\ttt,y)\in C^2(\bar Q),\  p(\r,\ttt,y)=p(\r,\ttt+2\pi,y),\  p(R_1,\ttt,y)=0\ for\ 0<\ttt\leq 2\pi,\ 0<y<L,
\end{equation*}

\begin{equation*}
 p(\r,\ttt,0)=0\ for\ R_1<\r<R_2,\ 0<\ttt \leq 2\pi,\  p(\r,\ttt,L)=0\ for\ R_1<\r<R_2,\ 0<\ttt \leq 2\pi\}.
\end{equation*}
Note that if $p(\r,\ttt,y)\in\cc$, $p(R_2,\r,y)$ is not prescribed.  Let $\bar p(\r,\ttt,y)\in \cc$ be a function which makes $J(p)$ a minimum in $\cc$ and let $v(\r,\ttt,y)$ be an arbitrary function of $\cc$ vanishing on the whole boundary of $Q$. If $\a\in \R$ we have $\bar p+\a v\in\cc$. Define $g_1(\e)=J(\bar p+\e v)$. Using the condition for minimality $g_1'(0)=0$ we find

\begin{equation}
\label{1_14}
\int_0^{2\pi}\int_{R_1}^{R_2}\int_0^L\biggl[p_\r(\r,\ttt,y)v_\r(\r,\ttt,y)+\frac{1}{\r^2}p_\ttt(\r,\ttt,y)v_\ttt(\r,\ttt,y)+p_y(\r,\ttt,y)v_y(\r,\ttt,y)\biggl]\r d\ttt d\r dy+
\end{equation}

\begin{equation*}
\int_0^{2\pi}\int_0^L\biggl\{ h^3(\ttt)\biggl[k_1 p_\ttt(R_2,\ttt,y)v_\ttt(R_2,\ttt,y)+k_2p_y(R_2,\ttt,y)v_y(R_2,\ttt,y)\biggl]-k_3 h'(\ttt)v(R_2,\ttt,y)\biggl\}R_2d\ttt dy=0.
\end{equation*}

Integrating by parts in (\ref{1_14}) and taking into account that $v$ vanishes on the boundary of $Q$, we find

\begin{equation}
\label{2_14}
\int_0^{2\pi}\int_{R_1}^{R_2}\int_0^L\biggl[p_{\r\r}(\r,\ttt,y)+\frac{1}{\r^2}p_{\ttt\ttt}(\r,\ttt,y)+p_{yy}(\r,\ttt,y)\biggl]v(\r,\ttt,y)\r d\ttt d\r dy=0.
\end{equation}
Since $v(\r,\ttt,y)$ is arbitrary we conclude that

\begin{equation}
\label{2_15}
p_{\r\r}(\r,\ttt,y)+\frac{1}{\r^2}p_{\ttt\ttt}(\r,\ttt,y)+p_{yy}(\r,\ttt,y)=0\ \ in\ \ Q
\end{equation}
thus equation (\ref{1_4_1})hols. Let us take now a different variation. More precisely let $w\in \cc$ and consider $g_2(\l)=J(\bar p+\l w)$. Recall that $w(\r,\ttt,y)$ is arbitrary on $C_2$. Computing $g'_2(0)=0$ we find again (\ref{1_14}), but with $w$ in place of $v$. This time, however, integrating by parts, the boundary term corresponding to $C_2$ does not vanish and we find, if we take into account (\ref{2_15}),
\begin{equation}
\label{2_16}
\int_0^{2\pi}\int_0^L\biggl[p_\r(R_2,\ttt,y)-\big(k_1 h^3(\ttt)p_\ttt(R_2,\ttt,y)\big)_\ttt-\big(k_2 h^3(\ttt)p_y(R_2,\ttt,y)\big)_y-
\end{equation}

\begin{equation*}
k_3h'(\ttt)\biggl]w(R_2,\ttt,y)R_2d\ttt dy=0.
\end{equation*}
 Given that $w(R_2,\ttt,y)$ is arbitrary we obtain

\begin{equation}
\label{3_16}
-k_1\bigl[h^3(\ttt)p_\ttt(R_2,\ttt,y)\bigl]_\ttt-k_2\bigl[h^3(\ttt)p_y(R_2,\ttt,y)\bigl]_y=-p_\r(R_2,\ttt,y)+k_3 h'(\ttt)
\end{equation}
 i.e.  (\ref{2_4_1}).  We use the established variational principle to prove that problem (\ref{1_4_1})-(\ref{6_4_1}) has one and only one solution. To this end, we  define a new space of admissible functions more mathematically convenient. Let $\hh(Q)$ be the set of functions $v(\r,\ttt,y)$ of class $C^\infty(\bar Q)$ periodic with period $2\pi$ with respect to $\ttt$ which take arbitrary values on $C_2$  and vanish in a neighbourhood of the remaining part of $Q$. We define on $\hh$ the norm

\begin{equation}
\label{1_18}
\pp v\pp=\biggl\{\int_0^{2\pi}\int_{R_1}^{R_2}\int_0^L\biggl[v_\r^2(\r,\ttt,y)+\frac{1}{\r^2}v_\ttt^2(\r,\ttt,y)+v_y^2(\r,\ttt,y)\biggl]\r d\ttt d\r dy+
\end{equation}

\begin{equation*}
\int_0^{2\pi}\int_0^L\biggl[k_1 v_\ttt^2(R_2,\ttt,y)+k_2v_y^2(R_2,\ttt,y)\biggl]R_2d\ttt dy\biggl\}^{1/2}.
\end{equation*}
Let $H^1_{00}(Q)$ be the completion of $\hh$ with respect to the norm (\ref{1_18}). To prove existence and uniqueness of  solutions of problem (\ref{1_4_1})-(\ref{6_4_1}) we use the Lax-Milgram lemma, see \cite{B}, which, for the sake of completeness, we quote below.

\begin{lemma}
Let $H$ be an Hilbert space and $a(u,v)$ a bilinear form defined in $H\times H$. Assume that there exist a constant $C$ such that 

\begin{equation}
\label{1_19}
|a(u,v)|\leq C\pp u\pp\ \pp v\pp\quad for\ all\ u\, v\in H\quad (continuity)
\end{equation}
and a constant $\a>0$ such that

\begin{equation}
\label{2_19}
|a(v,v)|\geq \a\pp v\pp^2\quad for\ all\ u\in H\quad (coerciveness).
\end{equation}

Then, if $f\in H'$, there exists a unique solution of the problem

\begin{equation}
\label{3_19}
u\in H,\quad a(u,v)=<f,v>\quad for\ all\ v\in H.
\end{equation}
Moreover, if $a(u,v)$ is symmetric, $u$ is characterised by the minimum property

\begin{equation}
\label{1_20}
\frac{1}{2}a(u,u)-<f,u>=\min_{v\in H}\biggl\{\frac{1}{2}a(v,v)- <f,v>\biggl\}.
\end{equation}
\end{lemma}

In our context lemma 2.1 is used as follows: $H=H^1_{00}(Q)$,

\begin{equation}
\label{1_21}
a(u,v)=\int_0^{2\pi}\int_{R_1}^{R_2}\int_0^L\bigg[u_\r(\r,\ttt,y)v_\r(\r,\ttt,y)+\frac{1}{\r^2}u_\ttt(\r,\ttt,y)v_\ttt(\r,\ttt,y)+u_y(\r,\ttt,y)v_y(\r,\ttt,y)\bigg]\r d\ttt d\r dy+
\end{equation}

\begin{equation*}
\int_0^{2\pi}\int_0^L h^3(\ttt)\bigg[k_1 u_\ttt(R_2,\ttt,y)v_\ttt(R_2,\ttt,y)+k_2p_y(R_2,\ttt,y)v_y(R_2,\ttt,y)\bigg]R_2d\ttt dy,
\end{equation*}

\begin{equation}
\label{2_21}
<f,v>=\int_0^{2\pi}\int_0^{L}\biggl[k_3 h'(\ttt)v(R_2,\ttt,y)\biggl]R_2d\ttt dy.
\end{equation}

Recalling the definition (\ref{1_18}) of the norm in $H^1_{00}(Q)$ and using the Cauchy-Schwartz inequality we find that (\ref{1_19}) is satisfied. Since $h(\ttt)=c(1+\e \cos(\ttt))\geq c(1-\e)>0$ also (\ref{2_19}) holds.  In the case at hand, (\ref{3_19}) is the weak formulation of problem (\ref{1_4_1})-(\ref{6_4_1}), thus this problem has one and only one weak solution. If the weak solution is sufficiently regular to permit the integration by parts leading to (\ref{2_15}) and (\ref{3_16}) we obtain a classical solution.

 In the same vein, i.e. in term of a weak formulation, we can treat the problem with cavitation. To this end use will be made of the following generalisation \cite{KS} of the Lax-Milgram lemma.

\begin{theorem}

Let $H$, $a(u,v)$ and $f$ be as in Lemma 2.1 and $\kk$ be a convex closed subset of $H$. Then there exists a unique $u$ such that
\begin{equation}
\label{1_22}
u\in\kk,\quad a(u,v-u)\geq<f,v-u>\quad for\ all\ v\in\kk.
\end{equation}
Moreover, if $a(u,v)$ is symmetric $u$ is characterised by the minimum property

\begin{equation}
\label{1_23}
\frac{1}{2}a(u,u)-<f,u>=\min_{v\in \kk}\biggl\{\frac{1}{2}a(v,v)- <f,v>\biggl\}.
\end{equation} 
\end{theorem}

We apply this theorem with $\kk=\big\{v\in H^1_{00}(Q),\ v\geq 0\ on\ C_2\big\}$. The space $H$, the bilinear form $a(u,v)$ and $f$ shall remain as prescribed in (\ref{1_21}) and (\ref{2_21}). We conclude that problem (\ref{1_4_1})-(\ref{6_4_1}) with (\ref{1_5_1}) has a unique weak solution.

\section{Reduction of the boundary value problem to a single two-point problem}

We treat in this Section the case of the porous long bearing. This leads to a one-dimensional Reynolds equation. Let $Q=\{(x,y);\ 0<x<\pi,\ 0<y<a\}$ represent the porous matrix. We arrive, for the determination of the pressure, to the following boundary value problem: to find $p(x,y)$ such that

\begin{equation}
\label{1_26}
p_{xx}+p_{yy}=0\quad in\ Q
\end{equation}

\begin{equation}
\label{2_26}
p(0,y)=0\quad for\ 0<y<a
\end{equation}

\begin{equation}
\label{3_26}
p(x,0)=0\quad for\ 0<x<\pi
\end{equation}

\begin{equation}
\label{4_26}
p(\pi,y)=0\quad for\ 0<y<a
\end{equation}

\begin{equation}
\label{5_26}
-\big(h^3(x)p_x(x,a)\big)_x+p_y(x,a)=h'(x),
\end{equation}
where $h(x)=1+\e\cos(2x),\quad 0<\e<1$.  We  have set, to simplify notations \footnote{Note that these constants can be reintroduced easily with minor notational changes} $k_1=1$, $k_3=1$ and $c=1$ in (\ref{5_26}).  The problem (\ref{1_26})-(\ref{5_26}) refers to the case without cavitation. To take cavitation into account we must add the unilateral condition $p(x,a)\geq 0$ with (\ref{5_26}) holding only where $p(x,a)>0$. We will show that the coupled problem (\ref{1_26})-(\ref{5_26}) can be reduced to a two-point problem for the unknown $P(x)=p(x,a)$. We first consider the following auxiliary problem: given $P(x)\in C^0([0,\pi])$ such that $P(0)=0$, $P(\pi)=0$, find $p(x,y)$ satisfying the Dirichlet's problem

\begin{equation}
\label{1_28}
p_{xx}+p_{yy}=0\quad in\ Q
\end{equation}

\begin{equation}
\label{2_28}
p(0,y)=0\quad for\ 0<y<a
\end{equation}

\begin{equation}
\label{3_28}
p(x,0)=0\quad for\ 0<x<\pi
\end{equation}

\begin{equation}
\label{4_28}
p(\pi,y)=0\quad for\ 0<y<a
\end{equation}

\begin{equation}
\label{5_28}
p(x,a)=P(x).
\end{equation}

We have

\begin{lemma}
The problem (\ref{1_28})-(\ref{5_28}) has one and only one solution given by

\begin{equation}
\label{1_29}
p(x,y)=\sum_{n=1}^\infty\frac{\sinh(ny)\sin(nx)}{\sinh(na)}b_n,
\end{equation}
where

\begin{equation*}
b_n=\frac{2}{\pi} \int_0^\pi \sin(n\xi)P(\xi)d\xi.
\end{equation*}
\end{lemma}

\begin{proof}
The Green's function  in $Q$ for the Laplacian is , see \cite{F} page 138,
$$
G(x,y;\xi,\ep)=
\left\{
\begin{array}{lr}
\sum_{n=1}^\infty\frac{2\sinh[n(a-y)]\sinh(n\ep)\sin(nx)\sin(n\xi)}{\pi n\sinh(na)}&if\ \ep\leq  y\\
\sum_{n=1}^\infty\frac{2\sinh(ny)\sinh[n(a-\ep)]\sin(nx)\sin(n\xi)}{\pi n\sinh(na)}&if\ y\leq \ep.
\end{array}
\right.
$$

Let $\nnn$ be the outward unit vector normal to the boundary of $Q$. On the boundary of $Q$ $p$ is different from zero only on the side $\{(x,a)\ 0<x<\pi\}$, where we have

\begin{equation}
\label{1_30}
\frac{\pa }{\pa n}=\frac{\pa }{\pa \ep}.
\end{equation}

Therefore

\begin{equation}
\label{2_30}
p(x,y)=-\int_0^\pi\frac{\pa G}{\pa \ep}(x,y;\xi,a)P(\xi)d\xi.
\end{equation}

On the other hand,

\begin{equation*}
\frac{\pa G}{\pa \ep}(x,y;\xi,a)=-2\sum_{n=1}^\infty \frac{\sinh(ny)\sin(nx)\sin(n\xi)}{\pi\sinh(na)}.
\end{equation*}

Setting

\begin{equation}
\label{0_32}
b_n=\frac{2}{\pi}\int_0^\pi \sin(n\xi)P(\xi)d\xi
\end{equation}
we obtain

\begin{equation}
\label{1_32}
p(x,y)=\sum_{n=1}^\infty\frac{\sinh(ny)\sin(nx)}{\sinh(na)}b_n
\end{equation}

i.e. (\ref{1_29}).
\end{proof}

\begin{example}
Consider the case $P(x)=\sin(2x)$. Since

\begin{equation}
\label{2_32}
\end{equation}
$$
\int_0^\pi\sin(mx)\sin(nx)dx=
\left\{
\begin{array}{lr}
\frac{\pi}{2}&if\ m=n\\
0&if\ m\neq n
\end{array}
\right.
$$
we find easily

\begin{equation}
\label{3_32}
p(x,y)=\frac{\sinh(2y)\sin(2x)}{\sinh(2a)}.\quad\Box
\end{equation}
\end{example}
\vskip .5cm

\noindent From (\ref{1_32}) we obtain

\begin{equation}
\label{1_33}
\frac{\pa p}{\pa y}(x,a)=\sum_{n=1}^\infty n\coth(na)b_n\sin(nx).
\end{equation}
We conclude that the solution of problem (\ref{1_26})-(\ref{5_26}) can be obtained in two steps. First we solve the following integro-differential equation

\begin{equation}
\label{2_33}
-\big(h^3(x)P'(x))'+\sum_{n=1}^\infty n\coth(na)b_n\sin(nx)=h'(x)
\end{equation}
with the boundary conditions

\begin{equation}
\label{3_33}
P(0)=0,\quad  P(\pi)=0.
\end{equation}
Once the solution $P(x)$ of problem (\ref{2_33}), (\ref{3_33}) is known the pressure in the porous matrix $p(x,y)$ is obtained solving the Laplace equation (\ref{1_28}) with the boundary conditions
\begin{equation}
\label{a}
 p(x,a)=P(x),\  p(x,0)=0,\  p(0,y)=0,\ p(\pi,y)=0. 
\end{equation}

\begin{theorem}
There exists one and only one solution in $H^1_0(0,\pi)$ to problem (\ref{2_33}), (\ref{3_33}).
\end{theorem}

\begin{proof}
Let

\begin{equation}
\label{1_34}
P(x)=\sum_{n=1}^\infty b_n \sin(nx)\in H^1_0(0,\pi),\quad   b_n=\frac{2}{\pi}\int_0^\pi P(\xi)\sin(n\xi)d\xi
\end{equation}

\begin{equation*}
V(x)=\sum_{n=1}^\infty c_n \sin(nx)\in H^1_0(0,\pi),\quad   c_n=\frac{2}{\pi}\int_0^\pi V(\xi)\sin(n\xi)d\xi.
\end{equation*}
We have

\begin{equation}
\label{2_34}
\int_0^\pi P'(x)V'(x)dx=\sum_{n=1}^\infty n^2 b_n c_n
\end{equation}

and

\begin{equation}
\label{1_35}
\pp V\pp_{H^1_0(0,\pi)}=\biggl(\sum_{n=1}^\infty n^2 c_n^2\biggl)^{1/2},\quad \pp V\pp_{L^2(0,\pi)}=\biggl(\sum_{n=1}^\infty c_n^2\biggl)^{1/2}.
\end{equation}
To write in weak form (\ref{2_33}), (\ref{3_33}) we multiply (\ref{2_33}) by $V(x)$. Integrating by parts over $(0,\pi)$ we have by (\ref{3_33})

\begin{equation*}
-\int_0^\pi \bigl(h^3(x)P'(x)\bigl)'V(x)dx=\int_0^\pi h^3(x)P'(x)V'(x)dx.
\end{equation*}
After simple calculations we obtain, recalling (\ref{2_32})

\begin{equation*}
\int_0^\pi \biggl[\sum_{n=1}^\infty n\coth(na)b_n\sin(nx)\biggl]\biggl[\sum_{k=1}^\infty c_k\sin(kx)\biggl]dx=\frac{\pi}{2}\sum_{n=1}^\infty n\coth(na)b_nc_n.
\end{equation*}
Thus we can rewrite problem (\ref{2_33}), (\ref{3_33}) as follows: to find $P(x)\in H^1_0(0,\pi)$ such that

\begin{equation}
\label{b}
\int_0^\pi h^3(x)P'(x)V'(x)dx+\frac{\pi}{2}\sum_{n=1}^\infty n\coth(na)b_nc_n=\int_0^\pi h'(x)V(x)dx,\ for\ all\ V(x)\in H^1_0(0,\pi).
\end{equation}
We use the Lax-Milgram lemma to prove that (\ref{b}) has one and only one solution. Define in $H_0^1(0,\pi)$ the bilinear form

\begin{equation}
\label{c}
a(P,V)=\int_0^\pi h^3(x)P'(x)V'(x)dx+\frac{\pi}{2}\sum_{n=1}^\infty n\coth(na)b_nc_n.
\end{equation}
Since $\coth(na)\leq\coth(a)$ for $n\geq 1$, we have, using the Wirtinger's inequality,

\begin{equation*}
\biggl|\sum_{n=1}^{\infty}n\coth(na)b_nc_n\biggl|\leq\coth(a)\sum_{n=1}^\infty n|b_n||c_n|\leq\coth(a)\biggl(\sum_{n=1}^\infty n^2b_n^2\biggl)^{1/2}\biggl(\sum_{n=1}^\infty c_n^2\biggl)^{1/2}\leq
\end{equation*}

\begin{equation*}
\coth(a)\pp P\pp_{H^1_0(0,\pi)}\pp V\pp_{L^2(0,\pi)}\leq \coth(a)\pp P\pp_{H^1_0(0,\pi)}\pp V\pp_{H^1_0(0,\pi)}.
\end{equation*}
Thus $a(P,V)$ is bounded , but it is also coercive. In fact, since $h(x)^3\geq (1-\e)^3$ we have

\begin{equation*}
a(P,P)=\int_0^\pi h^3(x)P'(x)^2dx+\sum_{n=1}^\infty n\coth(na)b_n^2\geq (1-\e)^3\pp P\pp^2_{H^1_0(0,\pi)}.
\end{equation*}
Since $\int_0^\pi h'(x)V(x)dx$ defines a linear functional on $H^1_0(0,\pi)$ we conclude that problem (\ref{b}) has one and only one solution.
\end{proof}

If we have the additional condition $P(x)\geq 0$, we must search $P(x)$ in the closed and convex set $\kk=\big\{P(x)\in H^1_0(0,\pi),\ P(x)\geq 0\big\}$. The weak formulation of the unilateral problem becomes

\begin{equation}
\label{1_39}
P(x)\in\kk,\quad  a(P,V-P)\geq \int_0^\pi h'(x)(V(x)-P(x))dx\quad  for\ all\ V\in\kk.
\end{equation}
This time we use Theorem 2.2 to solve the variational inequality (\ref{1_39}). To obtain the pressure in the matrix we consider, as in the non-unilateral case, the Dirichlet's problem  (\ref{1_28}), (\ref{4_28}) taking in (\ref{5_28}) the boundary value $P(x)$ given by the solution of (\ref{1_39}).
\vskip .3cm
 Particularly interesting is the case of the infinitely long bearing of small eccentricity ratio. The two-point problem (\ref{2_33}), (\ref{3_33}) becomes

\begin{equation}
\label{1_41}
-P''(x)+\sum_{n=1}^\infty n\coth(na)b_n\sin(nx)=\sin(2x)
\end{equation}

\begin{equation}
\label{2_41}
P(0)=0,\quad P(\pi)=0,
\end{equation}

\noindent where

\begin{equation}
\label{3_41}
b_n=\frac{2}{\pi}\int_0^\pi P(\xi)\sin(n\xi)d\xi.
\end{equation}

The problem (\ref{1_41})-(\ref{3_41}) has an explicit and particularly simple solution. Substituting 

\begin{equation}
\label{4_41}
P(x)=\sum_{n=1}^\infty b_n\sin(nx)
\end{equation}
in (\ref{1_41}) we have

\begin{equation}
\label{1_42}
\sum_{n=1}^\infty n^2 b_n\sin(nx)+\sum_{n=1}^\infty n\coth(na)b_n\sin(nx)=\sin(2x).
\end{equation}
The $b_n$'s entering in (\ref{1_42}) are easily computed if we multiply (\ref{1_42}) by $\sin(nx)$, $n=1,2...$ and integrate over $[0,\pi]$ with respect to $x$, recalling (\ref{2_32}) we find $b_n=0$ if $n\neq 2$ and

\begin{equation}
\label{2_42}
b_2=\frac{1}{2\bigl[2+\coth(2a)\bigl]}.
\end{equation}
Whence from (\ref{4_41}) we obtain

\begin{equation}
\label{3_42}
P(x)=\frac{\sin(2x)}{2\bigl[2+\coth(2a)\bigl]}
\end{equation}
with the pressure in $\bar Q$ given, in view of (\ref{1_32}), by

\begin{equation*}
p(x,y)=\frac{\sinh(2y)\sin(2x)}{2\sinh(2a)\bigl[2+\coth(2a)\bigl]}.
\end{equation*}
\vskip.5cm
\noindent The problem  for the non-porous case corresponding to (\ref{1_41}), (\ref{2_41}) is

\begin{equation}
\label{1_43}
-P''_0(x)=\sin(2x),\quad P_0(0)=0,\quad P_0(\pi)=0
\end{equation}

which has the simple solution

\begin{equation}
\label{2_43}
P_0(x)=\frac{\sin(2x)}{4}.
\end{equation}

The comparison between (\ref{3_42}) and (\ref{2_43}) shows that

\begin{equation}
\label{3_43}
|P_0(x)|> |P(x)|.
\end{equation}

If we assume the "half-Sommerfeld" condition to take into account of cavitation, we have from (\ref{3_42})

\begin{equation}
\label{1_43_1}
\end{equation}
$$
P(x)=
\left\{
\begin{array}{lr}
\frac{\sin(2x)}{2\bigl[2+\coth(2a)\bigl]}&if\ 0\leq x\leq\frac{\pi}{2}\\
0&if \frac{\pi}{2}<x\leq \pi.
\end{array}
\right.
$$

From (\ref{1_43_1}) we can draw the conclusion that, at least in the present simplified model, the pressure in porous bearings is always less than in the corresponding non-porous bearings and that the pressure increases as the porous matrix increases in size.

\section{Properties of the solution of a certain one-dimensional variational inequality}

Let $f(x)$ be a continuous function defined in $[0,L]$ satisfying

\begin{equation}
\label{1_45}
f(x)>0\quad in\ [0,L/2),\quad and\ f(x)<0\quad in\ (L/2,L]
\end{equation}

\begin{equation}
\label{2_45}
f(L-x)=-f(x)\quad in\ [0,L].
\end{equation}

We wish to study the variational inequality $u(x)\in\kk$

\begin{equation}
\label{3_45}
\int_0^L u'(x)(v'(x)-u'(x))dx\geq \int_0^L f(x)(v(x)-u(x))dx\quad for\ all\ v\in\kk,
\end{equation}
where

\begin{equation}
\label{4_45}
\kk=\big\{v(x)\in H^1_0(0,L),\ v(x)\geq 0\ in\ [0,L]\big\}.
\end{equation}

If $L=\pi$ and $f(x)=\sin(2x)$, (\ref{3_45}) is the variational inequality of the long non-porous bearing of small eccentricity ratio. The solution of (\ref{3_45}) exists, is unique and of class $C^1([0,L])$ by a general result of the theory \cite{KS}. There is a close link between (\ref{3_45}) and the following over-determined two-point problem: to find $\bar\xi\in(0,L)$ and $u(x)\in C^1([0,L])$ such that

\begin{equation}
\label{1_46}
-u''(x)=f(x),\quad x\in(0,\bar\xi)
\end{equation}

\begin{equation}
\label{2_46}
u(0)=0
\end{equation}

\begin{equation}
\label{3_46}
u(\bar\xi)=0
\end{equation}

\begin{equation}
\label{4_46}
u'(\bar\xi)=0.
\end{equation}

\begin{lemma}
Let $f(x)$ satisfy (\ref{1_45})-(\ref{2_45}). There exists one and only one solution $(\bar\xi$, $u(x))$ of problem (\ref{1_46})-(\ref{4_46}). Moreover,

\begin{equation}
\label{1_47}
\frac{L}{2}<\bar\xi<L
\end{equation}

\begin{equation}
\label{2_47}
u(x)=x\int_0^{\bar\xi}f(t)dt-\int_0^x (x-z)f(z)dz
\end{equation}
and

\begin{equation}
\label{3_47}
u(x)>0\quad in\ (0,\bar\xi).
\end{equation}
\end{lemma}

\begin{proof}
Integrating equation (\ref{1_46}) and using (\ref{4_46}) we find

\begin{equation}
\label{1_47_1}
u'(x)=\int_0^{\bar\xi} f(t)dt-\int_0^x f(t)dt.
\end{equation}

From (\ref{3_46}) we obtain the equation $F(\xi)=0$, where

\begin{equation}
\label{1_48}
F(\xi):=-\int_0^{\xi}\biggl[\int_0^z f(t)dt\biggl]dz+\xi\int_0^{\xi} f(t)dt=\int_0^{\xi} zf(z)dz
\end{equation}
which determines $\bar\xi$. On the other hand, $F(\frac{L}{2})>0$ by (\ref{1_45}) and $F(L)<0$. Hence the equation $F(\xi)=0$ has at least one solution $\bar\xi$. This solution is unique. In fact, $F(\xi)>0$ in $(0,\frac{L}{2})$, thus by (\ref{4_46}) we have $\bar\xi>\frac{L}{2}$. If, by contradiction, there is a second solution $\xi^*$ we have $F'(\xi)=\xi f(\xi)<0$ in $(\frac{L}{2},L)$, $F(\bar\xi)=0$ and $F(\xi^*)=0$, but this cannot happen by Rolle's theorem. Finally the inequality (\ref{3_47}) follows from (\ref{1_45}), (\ref{2_45}) and (\ref{2_47}).
\end{proof}

\begin{example}
If $L=\pi$ and $f(x)=\sin(2\pi)$ we find $\bar\xi$ using (\ref{1_48}). We find 

\begin{equation*}
F(\xi)=\frac{1}{4}\sin(2\xi)-\frac{\xi}{2}\cos(2\xi)=0
\end{equation*}
equation which in $(0,\pi)$ has one and only one solution i.e. $\bar\xi=2.246704729$. The corresponding $u(x)$ solution of problem (\ref{1_46})-(\ref{4_46}) is given by (\ref{2_47}).
\end{example}

\begin{lemma}
Let $u(x)$ and $\bar\xi$ be the solution of problem (\ref{1_46})-(\ref{4_46}). Then 
\begin{equation}
\label{q}
\end{equation}
$$
\bar u(x)=
\left\{
\begin{array}{lr}
u(x)&x\in[0,\bar\xi)\\
0&x\in[\bar\xi,L]
\end{array}
\right.
$$
solves the variational inequality (\ref{3_45}), (\ref{4_45}).
\end{lemma}

\begin{proof}
In view of (\ref{3_45}) and (\ref{q}) we can rewrite (\ref{3_45}) in the equivalent form

\begin{equation}
\label{2_51}
\int_0^{\bar\xi} u'(x)(v'(x)-u'(x))dx\geq \int_0^{\bar\xi}f(x)(v(x)- u(x))dx+\int_{\bar\xi}^L f(x)v(x)dx.
\end{equation}
Integrating by parts we have, recalling (\ref{1_46}), (\ref{2_46}) and (\ref{4_46}),

\begin{equation}
\label{e}
\int_0^{\bar\xi} u'(x)(v'(x)-u'(x))dx=-\int_0^{\bar\xi}u"(x)(v(x)-u(x))dx=\int_0^{\bar\xi} f(x)(v(x)-u(x))dx.
\end{equation}
Thus (\ref{2_51}) becomes

\begin{equation}
\label{f}
\int_0^{\bar\xi} f(x)(v(x)-u(x))dx\geq \int_0^{\bar\xi} f(x)(v(x)-u(x))dx+\int_{\bar\xi}^L f(x)v(x)dx
\end{equation}
which is certainly true since by (\ref{1_45}) and (\ref{1_47})
\begin{equation*}
\int_{\bar\xi}^L f(x)v(x)dx<0.
\end{equation*}
\end{proof}

\section{conclusion}
 Consider the problems without cavitation. Using Lemma 2.2 we proved  that there is one and only one weak solution of problem (\ref{1_4_1})-(\ref{6_4_1}). Can we prove that this weak solution is also a classical solution, in other words that the weak solutions has all the derivatives needed to satisfy the equations (\ref{1_4_1}) and (\ref{2_4_1}) ? One would expect even more, i.e. that the solution is of class $C^\infty$. Completely different is the situation for the problem in which cavitation is taken into account. In this case there is certainly a threshold of regularity. Probably one cannot go beyond the $C^{0,\a}(Q)$ regularity as found in a special case in \cite{SC}. Another question worthy considering is the shape of the region of cavitation, see \cite{Ci}. In the variational inequality considered in Section 4 the region of cavitation (``the coincidence set'') is an interval. One expect the same to be true also in the porous case. More generally, the region of cavitation is studied in the non porous and  two-dimensional case in \cite{Cc}. It would be interesting to find similar results also in the porous case of Section 2.

\bibliographystyle{amsplain}

\end{document}